\documentclass[runningheads]{llncs}
\usepackage{graphicx}
\usepackage{amsmath}
\usepackage{amssymb}
\newcommand{\comment}[1]{}

\newcommand{\commentOut}[1]{}
\newcommand{\beq}{\begin{equation}}
\newcommand{\eeq}{\end{equation}}
\def\EE{\hbox{I\kern-.1667em\hbox{E}}}

\newcommand{\uniform}{\mathcal{U}}

\newcommand{\empirical}{\mathcal{E}}
\newcommand{\expfam}{\mathcal{F}}
\newcommand{\actual}{\mathcal{G}}
\newcommand{\ovuniform}{\overline{\mathcal{U}}}

\newcommand{\ovempirical}{\overline{\mathcal{E}}}
\newcommand{\ovexpfam}{\overline{\mathcal{F}}}

\begin{document}
\title{$k$-Cut:
    A Simple Approximately-Uniform Method for Sampling
    Ballots in Post-Election Audits\thanks{Supported by Center for Science of Information (CSoI), an
NSF Science and Technology Center, under grant agreement CCF-0939370.}}
\titlerunning{$k$-Cut: Simple Approximate Sampling}
%
\author{Mayuri Sridhar \and Ronald L. Rivest}
\authorrunning{M. Sridhar and R.Rivest}
%
\institute{Massachusetts Institute of Technology, Cambridge MA 02139, USA\\
\email{mayuri@mit.edu}\\
\email{rivest@csail.mit.edu}}
\maketitle              
\begin{abstract}
  We present an approximate sampling framework and discuss how
  risk-limiting audits can compensate for these approximations, while
  maintaining their ``risk-limiting'' properties. Our framework is general
  and can compensate for counting mistakes made during audits.
  
  Moreover, we present and analyze a simple approximate sampling method,
  ``$k$-cut'', for picking a ballot randomly from a stack, without counting. Our method
  involves doing $k$ ``cuts,'' each involving moving a random portion
  of ballots from the top to the bottom of the stack,
  and then picking the ballot on top.  
  Unlike conventional methods of picking a ballot at
  random, $k$-cut does not require identification numbers on the
  ballots or counting many ballots per draw. 
  We analyze how close the distribution of chosen ballots is to
  the uniform distribution, and design different mitigation procedures.
  We show that $k=6$ cuts is enough for an risk-limiting election audit, based on
  empirical data,
  which would provide a significant increase
  in efficiency.

\keywords{sampling \and elections \and auditing \and post-election audits \and
              risk-limiting audit \and Bayesian audit.}
\end{abstract}
%
%
%
%
%

\section{Introduction}
\label{sec:introduction}

The goal of post-election tabulation audits are to provide assurance
that the reported results of the contest are correct; that is, they agree with
the results that a full hand-count would reveal.  To do this, the auditor
draws ballots uniformly at random one at a time from the set of all cast paper
ballots, until the sample of ballots provides enough assurance that the reported outcomes are correct.

The most popular post-election audit method is known as a ``risk-limiting
audit'' (or RLA), invented by Stark (see his web page~\cite{Stark-papers-talks-video}).
See also
\cite{Bretschneider-2012-risk,Johnson-2004-election-certification,Lindeman-2008-principles,Lindeman-2012-gentle,Rivest-2018-bayesian-tabulation-audits,Rivest-2012-bayesian}
for explanations, details, and related papers.
An RLA takes as input a ``risk-limit'' $\alpha$ (like $0.05$), and ensures that
if a reported contest outcome is incorrect, then this error will be
detected and corrected with probability at least $1-\alpha$.

This paper provides a novel method for drawing a sample of the cast
paper ballots.  The new method may often be more efficient than
standard methods.  However, it has a cost: ballots are drawn in a way
that is only ``approximately uniform''.
This paper also provides ways of compensating for such
non-uniformity.

There are two standard approaches for drawing a random sample of
cast paper ballots:
\begin{enumerate}
\item
  \textbf{[ID-based sampling]}
  Print on each scanned cast paper ballot a unique identifying number (ballot ID numbers).
  Draw a random sample of ballot ID numbers, and retrieve the corresponding ballots.
\item
  \textbf{[Position-based sampling]}
  Give each ballot an implicit ballot ID equal to its position in a canonical
  listing of all ballot positions.
  Then proceed as with method (1).
\end{enumerate}

These methods work well, and are guaranteed to produce random samples. 
In practice, auditors use software, like~\cite{Stark-2017-ballot-polling-tools}, which
takes in a ballot manifest as input and produces the random sample of ballot ID numbers.
In this software, it is typically assumed that sampling is done without replacement.

However, finding even a single ballot using these sampling methods can be tedious and
awkward in practice. For example, given a random sample of ID numbers, one may
need to count or search through a stack of ballots to find the desired
ballot with the right ID or at the right position. Moreover, typical auditing procedures assume
that there are no mistakes when finding the ballots for the sample. Yet, this seems
to be an unreasonable assumption - if we require a sample size of 1,000 ballots, for instance,
it is likely that there are a few ``incorrectly'' chosen ballots along the way, due to counting errors.
In the literature about RLAs, there is no way to correct for these mistakes.

\emph{Our goal is to simplify the sampling process.}

In particular, we define a general framework for compensating for ``approximate sampling''
in RLAs. Our framework of approximate sampling can be used to measure
and compensate for human error rate while using the counting methods outlined
above. Moreover, we also define a simpler approach for
drawing a random sample of ballots, which does not rely on counting at all. Our technique is simple
and easy to iterate on and may be of particular interest when the
stack of ballots to be drawn from is large. We define mitigation procedures to account
for the fact that the sampling technique is no longer uniformly random.

\paragraph{Overview of this paper.} 

Section~\ref{sec:notation} introduces the relevant notation that we use 
throughout the paper.

Section~\ref{sec:k-cut} presents our proposed sampling
method, called ``$k$-cut.'' 

Section~\ref{sec:single-ballot-selection}
studies the distribution of single cut sizes,
and provides experimental data.
We then show how iterating a single cut
provides improved uniformity for ballot selection.
%

Section~\ref{sec:approximate-sampling} discusses the major questions
that are brought up when using ``approximate'' sampling in a post-election audit.

%

Section~\ref{sec:auditing-arbitrary-contests} proves a very general
result: that any general statistical auditing procedure for an
arbitrary election can be adapted to work with
approximate sampling, with simple mitigation procedures.

Section~\ref{sec:multi-stack-sampling} discusses how to adapt the
$k$-cut method for sampling when the ballots are organized into
multiple stacks or boxes.

Section~\ref{sec:practical-guidance} provides some guidance
for using $k$-cut in practice.

Section~\ref{sec:discussion-and-open-problems} gives some 
further discussion, lists some
open problems, and makes some
suggestions for further research.

Section~\ref{sec:conclusions} summarizes our contributions.

\section{Notation and Election Terminology}
\label{sec:notation}

\paragraph{Notation.}
  

We let $[n]$ denote the set $\{0, 1, \ldots, n-1\}$, and
we let $[a, b]$ denote the set $\{a, a+1, \ldots, b-1\}$.

We let $\uniform[n]$ denote the uniform distribution over the
set $[n]$. In $\uniform[n]$, the ``$[n]$'' may be omitted when it is understood to be $[n]$,
where $n$ is the number of ballots in the stack. We let $\uniform[a, b]$ denote the uniform distribution over
the set $[a, b]$.  
If $X \,\sim\, \uniform[n]$, then
\[
    \\Pr[X = i] = \uniform[n](i) = 1/n \hbox{~for~} i \in [n]\ .
\]
 Thus,
$\uniform$ denotes the uniform distribution on $[n]$.
For the continuous versions of the uniform distribution: we let
$\ovuniform(0,1)$ denote the uniform distribution over the real interval
$(0,1)$, and let $\ovuniform(a,b)$ denote the uniform distribution over
the interval $(a, b)$. 
These are understood to be probability densities, not discrete distributions.
The ``$(0,1)$'' may be omitted when it is understood
to be $(0,1)$.  Thus, $\ovuniform$ denotes the uniform distribution on $(0,1)$.

We let $VD(p, q)$ denote the variation distance between probability
distributions $p$ and $q$; this is the maximum, over all events $E$,
of $$Pr_p[E]-Pr_q[E].$$

\paragraph{Election Terminology.}

The term ``ballot'' here means to a single piece of paper on which
the voter has recorded a choice for each contest for which the voter
is eligible to vote. One may refer to a ballot as a ``card.'' Multi-card ballots
are not discussed in this paper.
%

\paragraph{Audit types.}

There are two kinds of post-election audits: \emph{ballot-polling}
audits, and \emph{ballot-comparison} audits, as described in~\cite{Lindeman-2012-gentle}.
For our purposes, these types of audits are equivalent, since they both
need to sample paper ballots at random, and can make use of the
$k$-cut method proposed here.  However, if one wishes to use $k$-cut
sampling in a comparison audit, one would need to ensure that each
paper ballot contains a printed ID number that could be used to locate
the associated electronic CVR.


\section{The $k$-Cut Method}
\label{sec:k-cut}

The problem to be solved is:
\begin{quote}
  How can one select a single ballot (approximately) at random from a
  given stack of $n$ ballots?
\end{quote}

This section presents the ``$k$-cut'' sampling procedure for doing
such sampling. The $k$-cut procedure does not need to know the size $n$ of the stack,
nor does it need any auxiliary random number generators or technology.

We assume that the collection of ballots to be sampled from is in the
form of a stack.  These may be ballots stored in a single
box or envelope after scanning. One may think of the stack of ballots as being similar to
a deck of cards.  When the ballots are organized into \emph{multiple}
stacks, sampling is slightly more complex---see
Section~\ref{sec:multi-stack-sampling}.

For now we concentrate on the
single-stack case. We imagine that the size $n$ of the stack is
25--800 or so.

The basic operation for drawing a single ballot is called ``$k$-cut
and pick,'' or just ``$k$-cut.''  This method does $k$ cuts then draws
the ballot at the top of the stack.

To make a single cut of a given stack of $n$ paper ballots:
\begin{itemize}
\item Cut the stack into two parts: a ``top'' part and
  a ``bottom'' part.
\item Switch the order of the parts, so what was the bottom
  part now sits above the top part.  The relative order of the
  ballots within each part is preserved.
\end{itemize}
  
We let $t$ denote the size of the top part. The size $t$ of the top part
should be chosen ``fairly randomly'' from
the set $[n]=\{0, 1, 2, \ldots, n-1\}$\footnote{A cut of size $n$ is excluded,
as it is equivalent to a cut of size $0$.}.  In practice, 
cut sizes are probably not chosen so uniformly; so in this paper
we study ways to compensate for non-uniformity. We can also view the cut operation as one that ``rotates'' the stack of
ballots by $t$ positions.

\paragraph{An example of a single cut.}

As a simple example, if the given stack has $n=5$ ballots:
\[
    \framebox{A B C D E} \,,
\]
where ballot $A$ is on top and ballot $E$ is at the bottom, then a cut
of size $t=2$ separates the stack into a top part of size $2$
and a bottom part of size $3$:
\[
    \framebox{A B}~~\framebox{C D E} 
\]
whose order is then switched:
\[
    \framebox{C D E}~~\framebox{A B}\ .
\]
Finally, the two parts are then placed together to form the final stack:
\[
    \framebox{C D E A B} \,.
\]
having ballot $C$ on top.

\paragraph{Relative sizes}

We also think of cut sizes in relative manner, as a fraction of~$n$.
We let $\tau = t/n$ denote a cut size $t$ viewed as a fraction of the
stack size $n$.  Thus $0\le \tau < 1$.

\paragraph{Iteration for $k$ cuts.}

The $k$-cut procedure makes $k$ successive cuts then picks the ballot
at the top of the stack.

If we let $t_i$ denote the size of the $i$-th cut, then the net
rotation amount after $k$ cuts is

\begin{equation}
      r_k = t_1 + t_2 + \cdots + t_k \pmod{n}\ .
\end{equation}
The ballot originally in position $r_k$ (where the top ballot position
is position $0$) is now at the top of the stack. We show that even for
small values of $k$ (like $k=6$) the
distribution of $r_k$ is close to $\uniform$.

In relative terms, if we define
\[ \tau_i = t_i/n \]
and
\[ \rho_k = r_k/n\ , \]
we have that
\begin{equation}
      \rho_k = r_k / n = \tau_1 + \tau_2 + \cdots + \tau_k \pmod{1}\ .
\end{equation}

\paragraph{Drawing a sample of multiple ballots.}

To draw a sample of $s$ ballots, our
$k$-cut procedure repeats $s$ times the operation of drawing without
replacement a single ballot ``at random.''  The $s$ ballots so drawn
form the desired sample.


\paragraph{Efficiency.}

Suppose a person can make six (``fairly random'') cuts in approximately 15
seconds, and can count 2.5 ballots per second\footnote{These assumptions are based on empirical observations
during the Indiana pilot audits.}.  Then $k$-cut (with $k=6$) is
more efficient when the number of ballots that needs to be
counted is 37.5 or more. Since batch sizes in audits are
often large, $k$-cut has the potential to increase sampling speed.

For instance, assume that ballots are organized into boxes,
each of which contains at least 500 ballots. Then, when the counting
method is used, 85\% of the
time a ballot between ballot \#38 and ballot \#462 will be chosen. In
such cases, one must count at least~38 ballots from the bottom or from
the top to retrieve a single ballot. This implies that $k$-cut is more efficient
85\% of the time.

This analysis assumes that each time we retrieve a ballot, we start from the top
of the stack and count downwards. In fact, if we have to retrieve a single ballot
from each box, this is the best technique that we know of. However, let's instead
assume that we would like to retrieve $t$ ballots in each box of $n$ ballots These ballots are
chosen uniformly at random from the box; thus, in expectation, the largest ballot
position (the ballot closest to the bottom of the stack) will be $\frac{nt}{t+1}$.
One possible way to retrieve these $t$ ballots is to sort the required ballot IDs, by position,
and retrieve them in order, by making a single pass through the stack. This requires only
counting $\frac{nt}{t+1}$ ballots in total, to find all $t$ ballots. Using our estimate that a person
can count 2.5 ballots per second, this implies that if we sample $t$ ballots per box, each box
will require $\frac{nt}{2.5(t+1)}$ seconds. Using $k$-cut, we will require 15 seconds per draw,
and thus, $15t$ seconds in total.

This implies that $k$-cut is more efficient when
\begin{align*}
\frac{nt}{2.5(t+1)} &> 15t\,\\
n &> 37.5(t+1)\,.
\end{align*}

Thus, if we require 2 ballots per box ($t=2$), $k$-cut is more efficient, in expectation,
when there are at least 113 ballots per box. When $t=3$, then $k$-cut is more efficient, in expectation,
when there are at least 150 ballots per box. Since the batch sizes in audits are large,
and the number of ballots sampled per box is typically quite small, we expect $k$-cut
to show an increase in efficiency, in practice. Moreover, as the number of ballots per box increases,
the expected time taken by
standard methods to retrieve a single ballot increases. With
$k$-cut, the time it takes to select a ballot is
\emph{constant}, independent of the number of ballots in the box, assuming
that each cut takes constant time.

\paragraph{Security}
We assume that the value of $k$ is \textbf{fixed} in
advance; you can not allow the cutter to stop cutting once a ``ballot they like'' is
sitting on top.

%

\section{(Non)-Uniformity of Single Ballot Selection}
\label{sec:single-ballot-selection}

We begin by observing that if an auditor could perform
``perfect'' cuts, we would be done.  That is, if the auditor
could pick the size $t$ of a cut in a perfectly uniform
manner from $[n]$, then one cut would suffice
to provide a perfectly uniform distribution of the ballot
selected from the stack of size $n$. However, there is no
\emph{a priori} reason to believe that, even with
sincere effort, an auditor could pick $t$ in a perfectly uniform
manner.  

An auditor could pick a $t$ randomly from $[n]$ (or
pseudorandomly from $[n]$), and then count down in the stack until he
reach ballot~$t$.  But this ``counting down'' procedure is precisely
what we are trying to eliminate!

So, we start by studying the properties of the $k$-cut procedure for
single-ballot selection, beginning with a study of the non-uniformity
of selection for the case $k=1$ and extending 
our analysis to multiple cuts. 

\subsection{Empirical Data for Single Cuts}

This section presents our experimental data on single-cut sizes.
We find that in practice, single cut sizes (that is, for $k=1$) are
``somewhat uniform.''  We then show that the approximation to
uniformity improves dramatically as $k$ increases.

We had two subjects (the authors).  Each author had a stack of 150
sequentially numbered ballots to cut. Marion County, Indiana, kindly provided surplus ballots for us to work with.
The authors made 1680 cuts in total.
Table~\ref{table:combined} shows the observed cut size
frequency distribution.

\begin{table}
          \centering
          \begin{tabular}{|r||r|r|r|r|r|r|r|r|r|r||r|}          \hline
          &     0  &     1  &     2  &     3  &     4  &     5  &     6  &     7  &     8  &     9  & Row Sum  \\ 
         \hline\hline 
    0  &        0  &     0  &     0  &     2  &     3  &     5  &     5  &     6  &    10  &     7  &     38 \\ \hline
   10  &       10  &     6  &    11  &    12  &    16  &    10  &    11  &    16  &    12  &    16  &    120 \\ \hline
   20  &       21  &    22  &     7  &    18  &    25  &    15  &    25  &    21  &    18  &    16  &    188 \\ \hline
   30  &       16  &    23  &    15  &    20  &    19  &    19  &    15  &    16  &    20  &    20  &    183 \\ \hline
   40  &       18  &    17  &    22  &    24  &    12  &    17  &    17  &    20  &    25  &    28  &    200 \\ \hline
   50  &       16  &    13  &    17  &    17  &    17  &    20  &    14  &    16  &    27  &    13  &    170 \\ \hline
   60  &       15  &    17  &    14  &    13  &    14  &    14  &    13  &    13  &    17  &    16  &    146 \\ \hline
   70  &       10  &     9  &     8  &    10  &    14  &    16  &    14  &    21  &    25  &    11  &    138 \\ \hline
   80  &       13  &    11  &    11  &     5  &    14  &    14  &    14  &     8  &    15  &    12  &    117 \\ \hline
   90  &       13  &     9  &    17  &    19  &    10  &     6  &    14  &     6  &     2  &     4  &    100 \\ \hline
  100  &       12  &     8  &    10  &     8  &     5  &    10  &     6  &    11  &     9  &     9  &     88 \\ \hline
  110  &        4  &     9  &     9  &     8  &     4  &     9  &     6  &     9  &     7  &     9  &     74 \\ \hline
  120  &       10  &     7  &     6  &     5  &     4  &     6  &     8  &     5  &     6  &     3  &     60 \\ \hline
  130  &        4  &     4  &     8  &     4  &     6  &     0  &     4  &     6  &     2  &     4  &     42 \\ \hline
  140  &        1  &     3  &     2  &     4  &     0  &     2  &     3  &     0  &     0  &     1  &     16 \\ [1ex]

          \hline
          \end{tabular}
\caption{Empirical distribution of sizes of single cuts,
using combined data from both authors, with 1680 cuts total.
For example, ballot 3 was on top twice after one cut. (Note that the
initial top ballot is ballot 0.)}
\label{table:combined}
\end{table}

%

If the cuts were truly random, we would expect a
uniform distribution of the number of cuts observed as a function of
cut size. In practice, the frequency of cuts was not evenly distributed; there
were few or no very large or very small cuts, and smaller
cuts were more common than larger cuts.
  
\subsection{Models of Single-Cut Selection}

Given the evident non-uniformity of the single-cut sizes in our
experimental data, it is of interest to model their distribution.
Such models allow generalization to other stack sizes, and support the
study of convergence to uniformity with iterated cuts. In Figure~\ref{plot:empirical_model},
we can observe the probability density of the empirical distribution, compared
to different models.

\begin{figure}
\centering
\includegraphics[width=10cm, height=6cm,keepaspectratio]{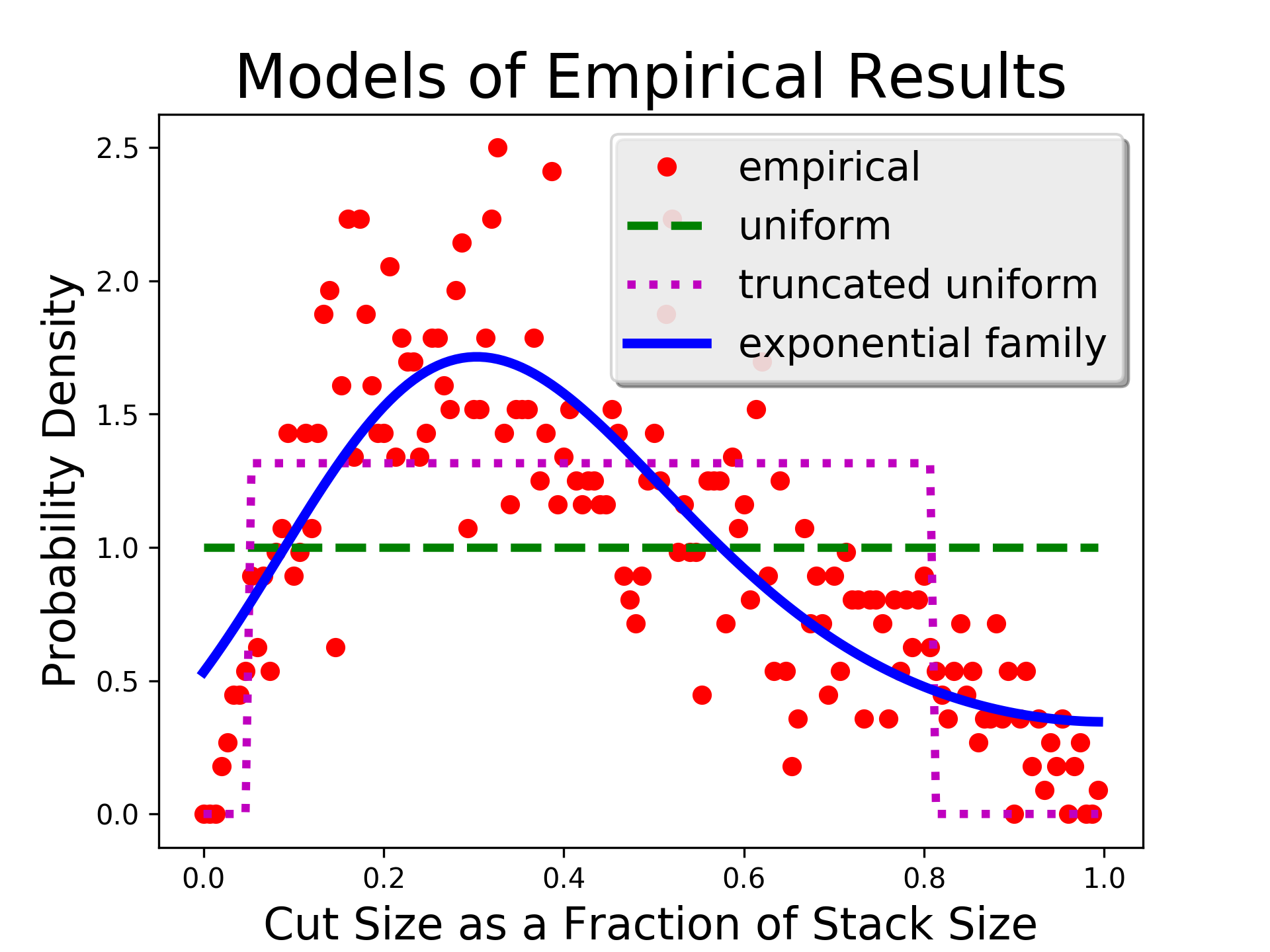}
\caption{Two models for cut sizes for a single cut, based on
the data of Table~\ref{table:combined}.
The horizontal axis is the size of the cut $\tau$ as a fraction of the
size of the stack of ballots.
The vertical axis is the probability density at that point.
For reference the uniform density $\ovuniform$ (shown in green) has a
constant value of $1$.
The red dots show the empirical distribution being modeled, which
is clearly not uniform.
The purple line shows our first model: the truncated uniform density 
$\ovuniform(0.533, 0.813)$
on the interval $8/150 \le \tau < 122/150$.
This density has a mean absolute error of 0.384
compared to the empirical density, 
and a mean squared error of 0.224.
The blue line shows our second model:
the density function from the model $\ovexpfam$ of
equation~(\ref{eqn:F}), which fits a bit better, giving
a mean absolute error of 0.265
and a mean squared error of 0.114.
}
\label{plot:empirical_model}
\end{figure}

We let $\empirical$ denote the observed empirical distribution on
$[n]$ of single-cut sizes, and let $\ovempirical$ denote the corresponding
induced continuous density function on $(0,1)$, of relative cut sizes.

We consider two models of the probability distribution of cut sizes
for a single cut.  

The reference case (the ideal case) is \textbf{the uniform model},
where $t$ is chosen uniformly at random from $[n]$ for the discrete
case, or when $\tau$ is chosen uniformly at random from the real
interval $(0, 1)$ for the continuous case.  We denote these cases as
$t \,\sim\, \uniform[n]$ or $\tau \,\sim\, \ovuniform(0, 1)$,
respectively.

We can define two different non-uniform models to reflect the observed data.

\begin{itemize}

\item \textbf{The truncated uniform model.}
  This model has two parameters: $w$ (the least cut size possible
  in the model) and $b$ (the number of different possible cut sizes).
  The cut size $t$ is chosen uniformly at random
  from the set $[w, w+b]= \{w, w+1, \ldots, w+b-1\}$.
  We denote this case as
  $t \,\sim\, \uniform[w, w+b]$ (for the discrete version)
  or
  $\tau \,\sim\, \ovuniform(w/n, (w+b)/n)$
  (for the continuous version).

\item An \textbf{exponential family} model.
  Here the density of relative cut sizes is modeled as
  $\mathcal{F}(\tau) = \exp(f(\tau))$,
  where $\tau$ is the relative cut size and
  $f$ is a polynomial (of degree three in our case).
\end{itemize}

\paragraph{Fitting models to data.}

We used standard methods to find least-squares best-fits
for the experimental data of Table~\ref{table:combined}
to models 
from the truncated uniform family and
from the exponential family based on cubic polynomials.

\paragraph{Fitted model - truncated uniform distribution.}

We find that choosing $w=8$ and $b=114$ provides the best
least-squares fit to our data.  This corresponds to a uniform
distribution
$t \,\sim\, \uniform[8,122]$
or
$\tau \,\sim\, \ovuniform(0.00667, 0.813)$.

\paragraph{Fitted model - exponential family.}

Using least-squares methods, we found a model from the exponential
family for the probability density of relative cut sizes for a single
cut, based on an exponential of a cubic polynomial of the relative cut
size $\tau$.

The model defines
\begin{equation}
   f(\tau) = -0.631 + 8.587 \tau - 18.446 \tau^2 + 9.428\tau^3\ 
\label{eqn:exponential-f}
\end{equation}
and then uses
\begin{equation}
   \ovexpfam(\tau) = \exp(f(\tau))
\label{eqn:F}
\end{equation}
as the density function of the exponential family function 
defined by $f$.
We can see in Figure~\ref{plot:empirical_model} that
this seems to fit our empirical observations quite well.

\subsection{Making $k$ successive cuts to select a single ballot}
\label{sec:making-k-successive-cuts}

As noted, the distribution of cut sizes for a single cut is noticeably
non-uniform. Our proposed $k$-cut procedure addresses this by iterating the
single-cut operation $k$ times, for some small fixed integer $k$.
This section discusses the
iteration process and its consequences. In later sections, we
discuss applying further mitigation methods to handle
any residual non-uniformity.

We assume for now that successive cuts are independent.
Moreover, we assume that sampling is done with replacement, for simplicity.
Using these assumptions, we provide computational results showing that as the number of cuts
increases, the $k$-cut procedure selects ballots with a distribution
that approaches the uniform distribution, for our empirical data, as well as our
fitted models. We compare by computing the
variation distance of the $k$-cut distribution from $\uniform$ for various $k$. 
We also computed $\epsilon$, the maximum
ratio of the probability of any single ballot under the empirical
distribution, to the probability of that ballot under the uniform
distribution, minus one\footnote{In Section~\ref{sec:empirical-support}, we
discuss why this value of $\epsilon$ is relevant}. Our results are summarized in
Table~\ref{table:vd-and-eps-versus-k}.

\begin{table}[!htbp]
\centering
\begin{tabular}{|c||c|c|c||c|c|c|} 
 \hline
     &  \multicolumn{3}{|c||}{Variation Distance} & 
        \multicolumn{3}{|c|}{Max Ratio minus one} \\ 
 \hline
 $k$ &
 $\empirical_k$ &
 $\uniform_k[w,w+b]$ &
 $\expfam_k$ &
 $\empirical_k$ &
 $\uniform_k[w,w+b]$ &
 $\expfam_k$
 \\
   \hline\hline
      1 &           0.247 &            0.24 &           0.212  &             1.5 &           0.316 &           0.707 \\
      2 &          0.0669 &          0.0576 &          0.0688  &           0.206 &           0.316 &           0.212 \\
      3 &          0.0215 &          0.0158 &          0.0226  &          0.0687 &          0.0315 &          0.0706 \\
      4 &          0.0069 &         0.00444 &         0.00743  &          0.0224 &          0.0177 &          0.0233 \\
      5 &         0.00223 &         0.00126 &         0.00244  &         0.00699 &         0.00311 &         0.00767 \\
      \textbf{6} &        \textbf{0.000719} &        \textbf{0.000357} &        \textbf{0.000802}  &         \textbf{0.00225} &         \textbf{0.00128} &         \textbf{0.00252} \\
      7 &        0.000232 &        0.000102 &        0.000264  &        0.000729 &        0.000284 &        0.000828 \\
      8 &        7.49e-05 &        2.92e-05 &        8.67e-05  &        0.000235 &        9.87e-05 &        0.000272 \\
      9 &        2.42e-05 &        8.35e-06 &        2.85e-05  &        7.59e-05 &        2.47e-05 &        8.95e-05 \\
     10 &        7.79e-06 &        2.39e-06 &        9.36e-06  &        2.45e-05 &        7.83e-06 &        2.94e-05 \\
     11 &        2.52e-06 &        6.86e-07 &        3.08e-06  &         7.9e-06 &        2.09e-06 &        9.67e-06 \\
     12 &        8.12e-07 &        1.97e-07 &        1.01e-06  &        2.55e-06 &        6.32e-07 &        3.18e-06 \\
     13 &        2.62e-07 &        5.64e-08 &        3.32e-07  &        8.23e-07 &        1.74e-07 &        1.04e-06 \\
     14 &        8.45e-08 &        1.62e-08 &        1.09e-07  &        2.66e-07 &        5.14e-08 &        3.43e-07 \\
     15 &        2.73e-08 &        4.63e-09 &        3.59e-08  &        8.57e-08 &        1.44e-08 &        1.13e-07 \\
     16 &         8.8e-09 &        1.33e-09 &        1.18e-08  &        2.77e-08 &         4.2e-09 &        3.71e-08 \\[1ex]
 \hline
\end{tabular}
\caption{
  Convergence of $k$-cut to uniform with increasing $k$.
  Variation distance from uniform and $\epsilon$-values for~$k$ cuts, as a function
  of~$k$, for $n=150$, where~$\epsilon$ is one less than the maximum ratio of the probability of
  selecting a ballot under the assumed distribution to the probability
  of selecting that ballot under the uniform distribution.
  The second through seventh column headings describe
  probability distribution of single-cut sizes convolved
  with themselves~$k$ times to obtain the~$k$-th row.
  Columns two and five give results for the
  distribution $\empirical_k$ equal to the 
  $k$-fold iteration of single cuts that have the distribution of
  the empirical data of Table~\ref{table:combined}.
  Columns three and six gives results 
  for the distribution $\uniform_k[w,b]$ equal to the
  $k$-fold iteration of single cuts that have the distribution
  $\uniform[8,122]$ that is the best fit of this class to the
  empirical distribution $\empirical$.
  Columns four and seven gives results
  for the distribution $\expfam_k$ equal to the 
  $k$-fold iteration of single cuts that have the distribution
  described in equations~(\ref{eqn:exponential-f}) and~(\ref{eqn:F}).
  The row for $k=6$ is bolded, since we will show that with our mitigation procedures,
  6 cuts is ``close enough'' to random.}
\label{table:vd-and-eps-versus-k}
\end{table}

We can see that, after six cuts, we get a variation distance of
about $7.19 \times 10^{-4}$, for the empirical distribution,
which is often small enough to justify our recommendation that
six cuts being ``close enough" in practice, for any RLA.

\subsection{Asymptotic Convergence to Uniform with $k$}
\label{sec:asymptotic-convergence}

As $k$ increases, the distribution of cut sizes provably approaches the uniform distribution, under
mild assumptions about the distribution of cut sizes for a single cut
and the assumption of independence of successive cuts.

This claim is plausible, given the analysis of similar situations for
continuous random variables.  For example, Miller and
Nigrini~\cite{Miller-2007-modulo-1-CLT} have analyzed the summation of
independent random variables modulo~$1$, and given necessary and
sufficient conditions for this sum to converge to the uniform
distribution.

For the discrete case, one can show that if once $k$ is large enough
that every ballot is selected by $k$-cut with some positive
probability, then as $k$ increases the distribution of cut sizes for
$k$-cut approaches $\uniform$.
Furthermore, the rate of
convergence is exponential.  The proof details are omitted here; however, the second
claim uses Markov-chain arguments, where each rotation amount is a
state, and the fact that the transition matrix is doubly stochastic.

\section{Approximate Sampling}
\label{sec:approximate-sampling}

We have shown in the previous section that as we iterate our $k$-cut
procedure, our distribution becomes quite close to the uniform distribution.
However, our sampling still is not exactly uniform. 

The literature on post-election audits generally assumes that
sampling is perfect. One exception is the paper by Banuelos and
Stark~\cite{Banuelos-2012-zombies}, which suggests dealing
conservatively with the situation when one can not find a ballot in an
audit, by treating the missing ballot as if it were a vote for the
runner-up. Our proposed mitigation procedures are similar in flavor.

In practice, sampling for election audits is often done using software
such as that by Stark~\cite{Stark-2017-ballot-polling-tools} or
Rivest~\cite{Rivest-2011-sampler}. Given a random
seed and a number $n$ of ballots to sample from, they can generate a
pseudo-random sequence of integers from $[n]$, indexing into a list of ballot
positions or ballot IDs. It is reasonable to treat such cryptographic sampling methods as
``indistinguishable from sampling uniformly,'' given the strength of
the underlying cryptographic primitives.

However, in this paper we deal with sampling that is not perfect;
the $k$-cut method with $k=1$ is obviously non-uniform, and even
with modest $k$ values, as one might use in practice, there will be
some small deviations from uniformity.

Thus, we address the following question:
\begin{quote}
  How can one effectively use an approximate sampling procedure in
  a post-election audit?
\end{quote}

We let $\actual$ denote the actual (``approximate'') probability
distribution over~$[n]$ from the sampling method chosen for the
audit. Our analyses assume that we have some bound on how close $\actual$ is
to $\uniform$, like variation distance. Furthermore, the quality of the approximation may be controllable,
as it is with $k$-cut: one can improve the closeness to uniform by
increasing $k$. We let $\actual^s$ denote the distribution on $s$-tuples of ballots
from $[n]$ chosen with replacement according to the distribution
$\actual$ for each draw.



%
%

\section{Auditing Arbitrary Contests}
\label{sec:auditing-arbitrary-contests}

This section proves a general result: for auditing an arbitrary
contest, we show that \emph{any}
risk-limiting audit can be adapted to work with approximate sampling,
if the approximate sampling is close enough to uniform. In particular, any RLA
can work with the $k$-cut method, if $k$ is large enough.

We show that if $k$ is sufficiently large, the resulting
distribution of $k$-cut sizes will be so close to uniform that any
statistical procedure cannot efficiently distinguish between the two.
That is, we want to choose $k$ to guarantee that
$\uniform$ and $\actual$ are close enough, so that any statistical
procedure behaves similarly on samples from each.

Previous work done by Baign\`eres in
\cite{BaigneresVaudenay-2008-Complexity-of-distinguishing-distributions}
shows that, there is an optimal distinguisher between two finite
probability distributions, which depends on the KL-Divergence between
the two distributions.

We follow a similar model to this work,
however, we develop a bound based on the variation distance between
$\uniform$ and $\actual$.

\subsection{General Statistical Audit Model}
\label{sec:general-statistical-audit-model}

We construct the following model, summarized in Figure \ref{thm-overview}.

\begin{figure}[!htbp]
\centering

 \includegraphics[width=10cm, height=6cm,keepaspectratio]{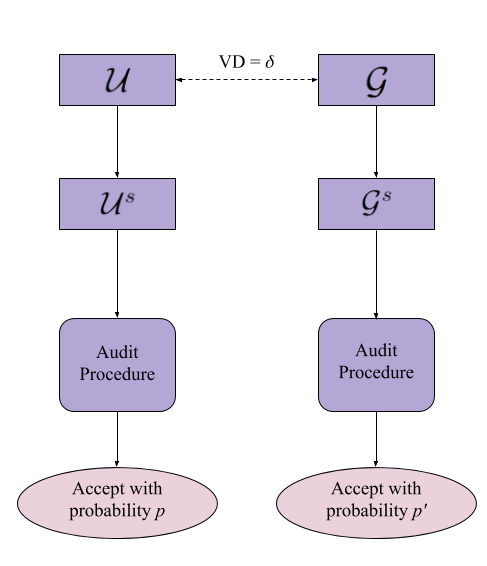}

\caption{Overview of uniform vs. approximate sampling effects, for any
statistical auditing procedure.
The audit procedure can be viewed as a distinguisher between the two
underlying distributions.  If it gives significantly different results for the
two distributions, it can thereby distinguish between them. 
However, if $p$ and $p'$ are
extremely close, then the audit cannot be used as a distinguisher.
}
\label{thm-overview}
\end{figure}

We define $\delta$ to be the variation distance between $\actual$ and
$\uniform$.  We can find an upper bound for $\delta$ empirically, as seen in 
Table~\ref{table:vd-and-eps-versus-k}.
If $\actual$ is the distribution of $k$-cut, then by increasing
$k$ we can make $\delta$ arbitrarily small.

The audit procedure requires a sample of some given size $s$, from $\uniform^{s}$ or $\actual^{s}$.
We assume that all audits behave deterministically. We do not
assume that successive draws are independent, although we assume that each
cut is independent.


Given the size $s$ sample, the audit procedure can make a decision on
whether to accept the reported contest result, escalate the audit, or
declare an upset.

\subsection{Mitigation Strategy}
When we use approximate sampling, instead of uniform sampling,
we need to ensure that the ``risk-limiting'' properties of the RLAs are 
maintained. In particular, as described in~\cite{Lindeman-2012-gentle}, an RLA with a risk limit of $\alpha$
guarantees that with probability at least $(1-\alpha)$ the audit will
find and correct the reported outcome if it is incorrect. We want to
show that we can maintain this property, while
introducing approximate sampling.

Without loss of generality, we focus on the
probability that the audit accepts the reported contest
result, since it is the case
where approximate sampling 
may affect the risk-limiting properties.
We show that
$\actual$ and $\uniform$ are sufficiently close
when $k$ is large enough,
that the
difference between $p$ and $p'$, as seen in Figure \ref{thm-overview}, is small. 

We show a simple
mitigation procedure, for RLA plurality elections, to compensate for
this non-uniformity, that we denote
as \textbf{risk-limit adjustment.}
For RLAs, we can simply
decrease the risk limit $\alpha$ by $|p'-p|$ (or an upper bound on this)
to account for the difference.
This decrease in the risk limit can accommodate the risk that
the audit behaves incorrectly due to approximate sampling.

\subsection{How much adjustment is required?}

%
%
We assume we have an auditing procedure $\mathbb{A}$, which accepts samples and outputs ``accept'' or ``reject''.
We model approximate sampling with
providing $\mathbb{A}$ samples from a distribution $\actual$. For our analysis, we look at the empirical distribution
of cuts in Table \ref{thm-overview}. For uniform sampling, we provide $\mathbb{A}$ samples from $\uniform$.

We would like to show that the probability that $\mathbb{A}$ accepts an outcome incorrectly, given samples from $\actual$ is not
much higher than the probability that $\mathbb{A}$ accepts an incorrect outcome, given samples from $\uniform$. We denote
$\mathbb{B}$ as the set of ballots that we are sampling from. 
%

\begin{theorem}
Given a fixed sample size $s$ and the variation distance $\delta$, the maximum change in probability that $\mathbb{A}$
returns ``accept'' due to approximate sampling is at most $$\epsilon_1 + (1+n\delta)^{s'} - 1,$$ where $s'$ is the maximum number of ``successes''
seen in $s$ Bernoulli trials, where each has a success probability of $\delta$, with probability at least $1-\epsilon_1$.
\end{theorem}

\begin{proof}
We define $s$ as the number of ballots that we pull from the set of cast ballots, before deciding whether or not to accept the outcome
of the election. Given a sample size $s$, based on our sampling technique, we draw
$s$ ballots, one at a time, from $\actual$ or from $\uniform$.

We model drawing a ballot from $\actual$ as first drawing a ballot from $\uniform$; however, with probability $\delta$, we replace the ballot
we draw from $\uniform$ with a new ballot from $\mathbb{B}$ following a distribution $\mathbb{F}$. We make no further assumptions
about the distribution $\mathbb{F}$, which aligns with our definition of variation distance. When drawing from $\actual$, for any ballot $b \in \mathbb{B}$, we have probability
at most $\frac{1}{n} + \delta$ of drawing $b$.

When we sample sequentially, we get a length-$s$ sequence of ballot IDs, $S$, for each of $\actual$ and $\uniform$. 
Throughout this model, we assume that we sample with replacement, although similar bounds should hold for sampling without replacement, as well. We define $X$ as the list of indices in the sequence $S$ where both $\actual$ and $\uniform$ draw the same ballot,
in order. We define $Z$ as the list of indices where $\actual$ has ``switched'' a ballot after the initial draw.
That is, for a fixed draw, $\uniform$ might produce the sample sequence [1, 5, 29]. Meanwhile, $\actual$
might produce the sample sequence sequence [1, 5, 30]. For this example, $X$ = [0, 1] and $Z$ = [2].

We define the set of possible size-$s$ samples as the set $D$. 
We choose $s'$ such that for any given value $\epsilon_1$, 
the probability that $|Z|$ is larger than $s'$ is at most $\epsilon_1$. Using this set up, we can calculate an upper bound on the probability that
$\mathbb{A}$ returns ``accept''. In particular, given the empirical distribution, the probability that $\mathbb{A}$ returns ``accept'' for a deterministic
auditing procedure becomes
$$\Pr[\mathbb{A}\ accepts\ |\ \actual] = \sum_{S \in D} \Pr[\mathbb{A}\ accepts\ |\ S] * \Pr[draw\ S\ |\ \actual]\,.$$

Now, we note that we can split up the probability that we can draw a specific sample $S$ from the distribution $\actual$. We know that with high
probability, there are at most $s'$ ballots being ``switched''. Thus,
$$\Pr[\mathbb{A}\ accepts\ |\ \actual]$$
$$= \sum_{S \in D} \Pr[\mathbb{A}\ accepts\ |\ S] * \Pr[draw\ S\ |\ \actual, S\ has\ \leq s'\text{ ``switched'' ballots}] * \Pr[\text{S has }\leq s'\text{ ``switched'' ballots}] $$
$$+ \sum_{S \in D} \Pr[\mathbb{A}\ accepts\ |\ S] * \Pr[draw\ S\ |\ \actual, S\ has\ >s'\text{ ``switched'' ballots}] * \Pr[\text{S has }> s'\text{ ``switched'' ballots}]\,.$$

\noindent Now, we note that the second term is upper bounded by
$$\Pr[\text{any size-s sample has more than }s'\text{ switched ballots}]\,.$$

\noindent We define the probability that any size-$s$ sample contains more than $s'$ switched ballots as $\epsilon_1$.

We note that, although the draws aren't independent, from the definition of variation distance, this is upper bounded by the probability
that a binomial distribution, with $s$ draws and $\delta$ probability of success.

Now, we can focus on bounding the first term.
We know that 
$$\Pr[\mathbb{A}\ accepts\ |\ \actual, \text{any sample has at most }s'\text{ switched ballots}]$$
$$= \sum_{S \in D} \Pr[\mathbb{A}\ accepts\ |\ S] * \Pr[draw\ S\ |\ \actual, S\ has\ \leq s'\text{ ``switched'' ballots}]$$

\noindent Meanwhile, for the uniform distribution, we know that the probability of accepting becomes
$$\Pr[\mathbb{A}\ accepts\ |\ \uniform] = \sum_{S \in D} \Pr[\mathbb{A}\ accepts\ |\ S] * \Pr[draw\ S\ |\ \uniform]\,.$$

\noindent Thus, we know that the change in probability becomes
$$\Pr[\mathbb{A}\ accepts\ |\ \actual] - \Pr[\mathbb{A}\ accepts\ |\ \uniform]\,$$
$$\leq \epsilon_1 + \sum_{S \in D} \Pr[\mathbb{A}\ accepts\ |\ S] (\Pr[draw\ S\ |\ \actual, S\ has\ \leq s'\text{ ``switched'' ballots}] - \Pr[draw\ S\ |\ \uniform])\,.$$

However, for any fixed sample $S$, we know that we can produce $S$ from $E$ in many possible ways. That is, we know that we have to draw at least $s-s'$
ballots that are from $\uniform$. Then, we have to draw the compatible $s'$ ballots from $\actual$. 
\noindent  In general, we define
the possible length $s-s'$ compatible shared list of indices as the set $\mathbb{X}$. That is, by conditioning on $\mathbb{X}$, we are now defining the exact
indices in the sample tally where the uniform and empirical sampling can differ. We note that $|\mathbb{X}| = {s \choose s'}$ and each possible set happens with equal probability.
Then, for any specific $x \in \mathbb{X}$, we can define $z$ as the remaining indices, which are allowed to differ from uniform
and approximate sampling. That is, if there are 3 ballots in the sample, and $x=[0, 1]$, then $z=[2]$. 

We can now calculate the probability that we draw some specific size-$s$ sample $S$, given the empirical distribution, and a fixed value of $s'$. 

$$\Pr[draw\ S\ |\ \actual] =  \sum_{x \in \mathbb{X}} \Pr[draw\ x\ |\ \uniform] * \Pr[draw\ z\ |\ \actual] * \Pr[\text{switched ballots are at indices in }z]$$

However, we know that for each ballot $b$ in $z$, we draw ballot $b$ with probability at most $\frac{1}{n} + \delta$, or $\frac{1+n\delta}{n}$. That is, for any ballot in $x$,
we know that we draw it with uniform probability exactly. However, for a ballot $b$ in $z$, we know that this a ballot that may have been ``switched''. In particular, with
probability $\frac{1}{n}$, we draw the correct ballot from $\uniform$. However, in addition to this, with probability $\delta$, we replace it with a new ballot - we assume that we
replace it with the correct ballot with probability 1. Thus, with probability at most $\frac{1}{n} + \delta$, we draw the correct ballot for this particular slot.  Thus, we get

\begin{align*}
& \Pr[draw\ S\ |\ \actual]\,\\
&=  \sum_{x \in \mathbb{X}} \Pr[draw\ x\ |\ \uniform] * \Pr[draw\ z\ |\ \actual] * \Pr[\text{switched ballots are at indices in }z]\,\\
&\leq  \sum_{x \in \mathbb{X}} \Pr[draw\ x\ |\ \uniform] * (\frac{1+n\delta}{n})^{s'} * \Pr[\text{switched ballots are at indices in }z]\,\\
&\leq (1+n\delta)^{s'} \sum_{x \in \mathbb{X}} \Pr[draw\ x\ |\ \uniform] * \Pr[draw\ z\ |\ \uniform]* \Pr[\text{switched ballots are at indices in }z]\,.
\end{align*}

Now, we note that there are ${s \choose s'}$ possible sequences $x \in \mathbb{X}$, where the ``switched'' ballots could be. Each of these possible sequences occurs with equal
probability, this becomes
\begin{align*}
&\Pr[draw\ S\ |\ \actual]\,\\
&\leq (1+n\delta)^{s'} \sum_{x \in \mathbb{X}} \Pr[draw\ x\ |\ \uniform] * \Pr[draw\ z\ |\ \uniform]* \Pr[\text{switched ballots are at indices in }z]\,.\\
&= (1+n\delta)^{s'} \sum_{x \in \mathbb{X}} \Pr[draw\ x\ |\ \uniform] * \Pr[draw\ z\ |\ \uniform] * \frac{1}{{s \choose s'}}\,\\
&= (1+n\delta)^{s'} \Pr[draw\ S\ |\ \uniform]\,.
\end{align*}

For an example, we can consider the sequence of ballots $S=$[1, 5, 29]. For simplicity, we assume that $s'=1$. Now, we would like to bound the probability that $\actual$ draws $S$.
We can split this up into cases:
\begin{enumerate}
\item $\actual$ produces [1, 5, 29] by drawing 1 and 5 from the uniform distribution, then drawing a ``switched ballot'' at slot 3, and drawing ballot 29, given the switched ballot at position 3.
\item $\actual$ produces [1, 5, 29] by drawing 1 and 29 from the uniform distribution, then drawing a ``switched ballot'' at slot 2, and drawing ballot 5, given the switched ballot at position 2.
\item $\actual$ produces [1, 5, 29] by drawing 5 and 29 from the uniform distribution, then drawing a ``switched ballot'' at slot 1, and drawing ballot 1, given the switched ballot at position 1.
\end{enumerate}
Thus, we define the possible compatible shared list of indices $\mathbb{X}$ as $$\mathbb{X} = \{[0,1], [1, 2], [0, 2]\}. $$ For each possible list $x \in \mathbb{X}$, we can define $z$ as the remaining
possible positions where we sample from $\actual$ instead of $\uniform$. That is, if $x = [0, 1]$, then $z = [2]$. In this case, we must first draw ballots 1 and 5 from the uniform distribution.
Then, assuming that the ballot at slot 2 can be switched, we know that with probability at most $\delta$, it is switched to take the value 29. With probability $\frac{1}{n}$, it takes the value 29 regardless.
Thus, the probability of generating the appropriate sample tally, given a possible switched ballot at slot 2, becomes $\frac{(1+n\delta)}{n^3}$, as desired.

Using this bound we can calculate our total change in acceptance probability. This becomes:
$$\Pr[\mathbb{A}\ accepts\ |\ \actual] - \Pr[\mathbb{A}\ accepts\ |\ \uniform]\,$$
\begin{align*}
&\leq \epsilon_1 + \sum_{S \in D} \Pr[\mathbb{A}\ accepts\ |\ S] (\Pr[draw\ S\ |\ \actual, S\ has\ \leq s'\text{ ``switched'' ballots}] - \Pr[draw\ S\ |\ \uniform])\,\\
&\leq \epsilon_1 + ((1+n\delta)^{s'} - 1) \sum_{S \in D} \Pr[\mathbb{A}\ accepts\ |\ S] \Pr[draw\ S\ |\ \uniform]\,\\
&\leq \epsilon_1 + (1+n\delta)^{s'} - 1\,,
\end{align*}
which provides us the required bound. 
\end{proof}

\subsection{Empirical Support}
\label{sec:empirical-support}

Our previous theorem gives us a total bound of our change in risk limit, which depends on our value of $s'$ and $\delta$. We note that, for each ballot $b$, we provide a general bound
of a multiplicative factor increase of $(1+n\delta)$, which is based off the variation distance of $\delta$. However, we note that in practice, the exact bound we are looking
for depends on the multiplicative increase in probability of a single ballot being chosen. That is, we can calculate the max increase in multiplicative ratio for a single ballot, compared
to the uniform distribution. Thus, if a ballot is chosen with probability at most $\frac{(1 + \epsilon_2)}{n}$, then our bound on the change in probability becomes
$$\epsilon_1 + (1+\epsilon_2)^{s'} - 1.$$
The values of $\epsilon_2$ are recorded, for varying number of cuts in Table~\ref{table:vd-and-eps-versus-k}.

We can calculate the maximum change in probability for a varying number of cuts using this bound. Here, we analyze the case of 6 cuts.
To get a bound on $s'$, we can model how often we switch ballots. In particular, this follows a binomial distribution, with $s$ independent
trials, where each trial has a $\delta_6$ probability of success. Using the binomial survival function, we see at most 4 ``switched ballots'' in 1,000 draws,
with probability (1- $8.78 \times 10^{-4}$). From our previous argument, we know that our change in acceptance probability is at most $(1+\epsilon_2)^4 - 1$. Using our
value of $\epsilon_2$ for $k=6$, this causes a change in probability of at most 0.0090.

Thus, the maximum possible change in probability of incorrectly accepting this outcome is $0.0090 + 8.78 \times 10^{-4}$, which is approximately $9.88 \times 10^{-3}$. We can compensate
for this by adjusting our risk limit by less than 1\%.

\section{Multi-stack Sampling}
\label{sec:multi-stack-sampling}


Our discussion so far presumes that all cast paper ballots constitute
a single ``stack,'' and suggest using our proposed $k$-cut procedure
is used to sample ballots from that stack. In practice, however, stacks have limited size, since
large stacks are physically awkward to deal with.  The collection of
cast paper ballots is therefore often arranged into multiple stacks  of some limited size.

The \emph{ballot manifest} describes this arrangement of ballots into
stacks, giving the number of such stacks and the number of ballots
contained in each one. We assume that the
ballot manifest is accurate. A tool like Stark's Tools for Risk-Limiting Audits~\footnote{\url{https://www.stat.berkeley.edu/~stark/Vote/auditTools.htm}}
takes the ballot manifest (together with a random seed and the desired
sample size) as input and produces a sampling plan.

A sampling plan describes exactly which ballots to pick from which
stacks.  That is, the sampling plan consists of a sequence of pairs,
each of the form: (stack-number, ballot-id), where ballot-id may be
either an id imprinted on the ballot or the position of the ballot
in the stack (if imprinted was not done).

Modifying the sampling procedure to use $k$-cut is straightforward.
We ignore the ballot-ids, and note only how many ballots are
to be sampled from each stack.  That number of ballots are then selected 
using $k$-cut rather than using the provided ballot-ids.

For example, if the sampling plan says that $2$ ballots are
to be drawn from stack $5$, then we ignore the ballot-ids for those
specific $2$ ballots, and return $2$ ballots drawn approximately
uniformly at random using $k$-cut.  

Thus, the fact that cast paper ballots may be arranged into
multiple stacks (or boxes) does not affect the usability of
$k$-cut for performing audits.

\section{Approximate Sampling in Practice}
\label{sec:practical-guidance}

The major question when using the approximate sampling procedure is
how to choose $k$.  Choosing a small value of $k$ makes the overall
auditing procedure more efficient, since you save more time in each
sample you choose. However, it requires more risk limit adjustment.
%

The risk limit mitigation procedure
requires knowledge of the maximum sample size, which we denote as $s^{*}$,
beforehand. We assume that the auditors have a reasonable procedure for estimating
$s^{*}$ for a given contest.  One simple procedure to estimate $s^{*}$ is to get an initial
small sample size, $s$, using
uniform random sampling. Then, we can use a statistical procedure to
approximate how many ballots we would need to finish the audit,
assuming the rest of the ballots in the pool are similar to the
sample. Possible statistical procedures which can be used here
include:
\begin{itemize}
\item Replicate the votes on the ballots,
\item Sample from the multinomial distribution, using the sample voteshares as hyperparameters,
\item Use the Polya's Urn model to extend the sample,
\item Use the workload estimate as defined in \cite{Lindeman-2012-bravo},
for a contest with risk limit $\alpha$ and margin $m$ to predict the number of samples required.
\end{itemize}

Let us assume that we use one of these techniques and calculate that
the audit is complete after an extension of size $d$. To be safe, we
suggest assuming that the required additional sample size for the audit is at
most $2d$ or $3d$, to choose the value of $k$. Thus, our final bound on $s^{*}$
would be $s+3d$. 

Given this upper bound, we can perform
our approximate sampling procedures and mitigation procedures, assuming that
we are drawing a sample of size $s^{*}$. If the sample size required is
greater than $s^{*}$, then the ballots which are sampled
after the first $s^{*}$ ballots should be sampled uniformly at random.


\paragraph{Use in Indiana pilot audit.}

On May 29--30, 2018, the county of Marion, Indiana held a pilot audit
of election results from the November 2016 general election.  This
audit was held by the Marion County Election Board with assistance
from the Voting System Technology Oversight Project (VSTOP) Ball State
University, the Election Assistance Commission, and the current
authors.

For some of the sampling performed in this audit, the ``Three-Cut''
sampling method of this paper was used instead of the ``counting to a
given position'' method.  The Three-Cut method was modified so that
three different people made each of the three cuts; the stack of
ballots was passed from one person to the next after a cut.

Although the experimentation was informal and timings were not
measured, the Three-Cut method did provide a speed-up in the sampling
process.

The experiment was judged to be sufficiently successful that further
development and experimentation was deemed to be justified.  (Hence
this paper.)


\section{Discussion and Open Problems}
\label{sec:discussion-and-open-problems}

We would like to do more experimentation on the variation between individuals
on their cut-size distributions. The current empirical results in this paper
are based off of the cut distributions of just the two authors in the paper. We would
like to test a larger group of people to better understand
what distributions should be used in practice.

After investigating the empirical distributions of cuts, we would like to develop
``best practices'' for using the $k$-cut procedure. That is, we'd like to
develop a set of techniques that auditors can use to produce
nearly-uniform single-cut-size distributions. This will make using the $k$-cut
procedure much more efficient.

Finally, we note that our analysis makes some assumptions about how $k$-cut
is run in practice. For instance, we assume that each cut is 
made independently. We would like to run some empirical experiments to test our
assumptions.

\section{Conclusions}
\label{sec:conclusions}

We have presented an approximate sampling procedure, $k$-cut, for use
in post-election audits. We expect the use of $k$-cut may save time
since it eliminates the need to count many ballots in a stack
to find the desired one.

We showed that as $k$ gets larger, our procedure provides a sample
that is arbitrarily close to a uniform random sample. Moreover, we showed
that even for small values of $k$, our procedure provides a sample that is close
to being chosen uniformly at random. We designed a simple mitigation procedure for RLAs
that accounts for any remnant non-uniformity, by adjusting the risk limit. Finally,
we provided a recommendation of $k=6$ cuts to use in practice,
for sample sizes up to 1,000 ballots, based on our empirical data.

\bibliographystyle{splncs04}

\bibliography{conf_paper}
\end{document}